\documentclass[a4paper,12pt]{article}
\usepackage{amssymb,amsthm,amsmath,amsfonts}
\usepackage{float}
\usepackage{cite}
\usepackage[pdftex]{graphicx}

\newtheorem{Theorem}{Theorem}
\newtheorem{remark}{Remark}
\newtheorem{definition}{Definition}
\newtheorem{Example}{Example}

\newcommand{\D}{{\bf D}}
\newcommand{\E}{{\bf E}}
\newcommand{\R}{{\mathbb R}}
\newcommand{\B}{{\bf B}}

\numberwithin{equation}{section}

\begin{document}

\title{{Biquaternionic Reformulation of a Fractional Monochromatic Maxwell System}}

\author{Yudier Pe\~na P\'erez$^{a_1}$, Ricardo Abreu Blaya$^{b_1}$,\\ Mart\'in Patricio \'Arciga Alejandre$^{b_2}$, Juan Bory Reyes$^{a_2}$.}
\date {\small{{$^a${SEPI-ESIME-ZAC-Instituto Polit\'ecnico Nacional, Ciudad M\'exico, M\'exico.\\
$^{a_1}${e-mail: ypenap88@gmail.com}, $^{a_2}${e-mail: juanboryreyes@yahoo.com}\\
$^b$Facultad de Matem\'atica, Universidad Aut\'onoma de Guerrero, Chilpancingo, Guerrero, M\'exico.\\
{$^{b_1}${e-mail: rabreublaya@yahoo.es} $^{b_2}$e-mail: mparciga@gmail.com}\\
}}}}

\maketitle

\begin{abstract} 
In this work we propose a biquaternionic reformulation of a fractional monochromatic Maxwell system. Additionally, some examples are given to illustrate how the quaternionic fractional approach emerges in linear hydrodynamic and elasticity. 

\vspace{0.3cm}

\small{
}  
\end{abstract}

\section{Introduction} 
The past few decades have witnessed a surge of interest in research on the theory of the Maxwell system. A technique to study the Maxwell system is to reduce it to the equivalent Helmholtz equation. In a series of recent papers diverse applications of the Maxwell system theory have been studied, see \cite{AAK, NA48, MSB} for more details.

The Dirac equation is an important one in mathematical physics used to represent the Maxwell system through several ways, which has attracted the attention of physicists and engineers, see \cite{VR}.

A new approach for the study of the Maxwell system by using the quaternionic displaced Dirac operator, rather than working directly with the Helmholtz equation, appeared recently.

The quaternionic analysis gives a tool of wider applicability for the study of electromagnetic problems. In particular, a quaternionic hyperholomorphic approach to monochromatic solutions of the Maxwell system is established in \cite{KS, A22}.

The fractional calculus goes back to Leibniz, Liouville, Grunwald, Letnikov and Riemann. There are many interesting books on this topic as well as in fractional differential equations, see e.g. \cite{OS, Samko, Miller, Podlubny, Kilbas, Her}. 

Nowadays, the fractional calculus is a progressive research area \cite{Baleanu1,Napoles2}. Among all the subjects, we mention the treatment of fractional differential equations regarding the mathematical methods of their solutions and their applications in physics, chemistry, engineering, optics and quantum mechanics. For more details we refer the reader to \cite{Gorenflo, Gorenflo2, Kilbas, Her, Miller, Podlubny, Zhou}.

The fractional derivative operators are non-local and this property is very important because it allows modeling the dynamic of many complex processes in applied sciences and engineering, see \cite{Bonilla, Paola}. For example, the fractional non-local Maxwell system and the corresponding fractional wave equations are considered in \cite{Tarasov2008,Baleanu,Trujillo1}.

Recently, Ferreira and Vieira \cite{Ferreira1} proposed a fractional Laplace and Dirac operator in $3$-dimensional space using Caputo derivatives with different orders for each direction. Previous approaches, but using Riemann-Liouville derivatives can be found in \cite{Yakubovich, Ferreira2}.

The main goal of this paper is to describe the very close connection between the $3$-parameter quaternionic displaced fractional Dirac operator using Caputo derivatives and a fractional monochromatic Maxwell system. 

After this brief introduction let us give a description of the sections of this paper. Section $2$ contains some basic and necessary facts about fractional calculus, fractional vector calculus and the connections between quaternionic analysis and fractional calculus. In Section $3$, we present some examples of fractional systems in Physics. Finally, Section $4$ is devoted to the study of a fractional monochromatic Maxwell system and summarize the main achievements of this study.

\section{Preliminaries}
In this section we introduce the fractional derivatives and integrals necessary for our purpose and review some standard facts on fractional vector calculus and basic definitions of quaternionic analysis.
 
\subsection{Fractional derivatives and integrals}
Definitions and results of fractional calculus are established in this subsection, see \cite{Samko, Kilbas, Podlubny}.
\begin{definition}[\cite{Samko}]
Let a real-valued function $f(x)\in L_1[a,b]$. The left Riemann-Liouville fractional integral of order $\alpha_1>0$ is given by
\begin{equation*}\label{IntegralFraccionaria}
\left(_aI^{\alpha_1}_{x}f\right)(x):=\frac{1}{\Gamma(\alpha_1)}\int_{a}^{x}\frac{f(\tau)}{(x-\tau)^{1-\alpha_1}}d\tau, \,\,\,\ x>a.
\end{equation*}
\end{definition}

\begin{definition}[\cite{Samko}]
The left Caputo fractional derivative of order $\alpha_1>0$ for $f(x)\in AC^1[a,b]$ is written as
\begin{equation*}\label{DerivadaCaputo2}
\left(_a^C\hspace{-0.05cm}D^{\alpha_1}_{x}f\right)(x):= \frac{1}{\Gamma(1-\alpha_1)}\int_{a}^{x}\frac{f'(\tau)}{(x-\tau)^{\alpha_1}}d\tau, \,\,\,\ 0<\alpha_1<1.
\end{equation*}
\end{definition}

Here and subsequently, $AC^1[a,b]$ denotes the class of continuously differentiable functions $f$ which are absolutely continuous on $[a,b]$.

It is easily seen that
\begin{equation}\label{DerivadaCaputo3}
\left(_a^C\hspace{-0.05cm}D^{\alpha_1}_{x}f\right)(x)= \left(_aI^{1-\alpha_1}_{x}f'\right)(x).
\end{equation}
Unfortunately, in general the semi-group property for the composition of Caputo fractional derivatives is not true. Conditions under which the law of exponents holds is established in the next theorem, which follows the main ideas proposed in \cite {Podlubny}.

\begin{Theorem}\label{th2}
Let $\alpha_1, \alpha_2\in(0,1]$ such that $\alpha_1+\alpha_2>1$ and $f\in C^{2}[a,b]$. Then
\begin{equation}\label{Semigrupo} 
\left(_a^C\hspace{-0.05cm}D^{\alpha_1}_{x} \hspace{0.1cm} _a^C\hspace{-0.05cm}D^{\alpha_2}_{x}f\right)(x)= \hspace{0.1cm}\left(_a^C\hspace{-0.05cm}D^{\alpha_1+\alpha_2}_{x}f\right)(x)
\end{equation}
holds if the function $f$ satisfies the condition 
\begin{equation}\label{Condition}
f'(a)=0.
\end{equation} 
\end{Theorem}   
 
\begin{proof}
Applying (\ref{DerivadaCaputo3}) yields 
\begin{equation*}\label{Semigrupo1}
\left(_a^C\hspace{-0.05cm}D^{\alpha_1}_{x} \hspace{0.1cm} _a^C\hspace{-0.05cm}D^{\alpha_2}_{x}f\right)(x) = \hspace{0.1cm}\left(_aI^{1-\alpha_1}_{x}(_a^C\hspace{-0.05cm}D^{\alpha_2}_{x}f)'\right)(x).
\end{equation*} 
From \cite{Podlubny} (p. 81) and (\ref{Condition}), it follows that 
\begin{eqnarray*}\label{Semigrupo2}
(_a^C\hspace{-0.05cm}D^{\alpha_2}_{x}f)'(x)&=& \left(_a^C\hspace{-0.05cm}D^{1+\alpha_2}_{x}f\right)(x).\end{eqnarray*} 
Consequently, 
\begin{equation*}\label{Semigrupo3}
\left(_a^C\hspace{-0.05cm}D^{\alpha_1}_{x} \hspace{0.1cm} _a^C\hspace{-0.05cm}D^{\alpha_2}_{x}f\right)(x)= \left(_aI^{1-\alpha_1}_{x}\hspace{0.1cm} {}_a^C\hspace{-0.05cm}D^{1+\alpha_2}_{x}f\right)(x).
\end{equation*} 
But $1+\alpha_2<2$, then
\begin{eqnarray*}\label{Semigrupo4}
\left(_a^C\hspace{-0.05cm}D^{\alpha_1}_{x} \hspace{0.1cm} _a^C\hspace{-0.05cm}D^{\alpha_2}_{x}f\right)(x)&=& \left(_aI^{1-\alpha_1}_{x}\hspace{0.1cm}_aI^{2-(1+\alpha_2)}_{x}f''\right)(x)\\\nonumber&=&\left(_aI^{1-\alpha_1}_{x}\hspace{0.1cm}_aI^{1-\alpha_2}_{x}f''\right)(x)\\\nonumber&=&\left(_aI^{2-(\alpha_1+\alpha_2)}_{x}f''\right)(x)\\\nonumber&=&(_a^C\hspace{-0.05cm}D^{\alpha_1+\alpha_2}_{x}f)(x).
\end{eqnarray*} 
\end{proof}

\subsection{Elements of quaternionic functions}

We follow Kravchenko \cite{A22} in assert that: \textit{“The whole building which the equations of mathematical physics inhabit can be erected on the foundations of quaternionic analysis, and this possibility represents some interest due to the lightness and transparency especially of the highest floors of that new building as well as due to high speed horizontal (apart from the vertical) movement allowing an extremely valuable communication between its different parts. Nevertheless the current major interest may be the tools of quaternionic analysis which permit results to be obtained where other more traditional methods apparently fail”}.

Let $\mathbb H(\R)$ be the skew field of real quaternions and let $e_0=1,e_1, e_2, e_3$ be the quaternion units that fulfill the condition  
\[e_me_n+e_ne_m=-2\delta_{mn},\;m,n=1,2,3\]
\[e_1e_2=e_3;\;e_2e_3=e_1;\;e_3e_1=e_2.\]

Let $q=q_0+\vec{q}=\sum_{n=0}^3q_ne_n$, where $q_0 =: Sc(q)$ is called scalar part and $\vec{q}=: Vec(q)$ is called vector part of the quaternion $q$. The conjugate element $\bar{q}$ is given by $\bar{q}=q_0-\vec{q}$. If $Sc(q)=0$ then $q=\vec{q}$ is called a purely vectorial quaternion and it is identified with a vector $\vec{q}=(q_1,q_2,q_3)$ from $\R^3$.

The multiplication of two quaternions $p, q$ can be rewritten in vector terms:
\[pq=p_0q_0-\vec{p}\cdot\vec{q}+p_0\vec{q}+q_0\vec{p}+\vec{p}\times\vec{q},\]
where $\vec{p}\cdot\vec{q}$ and $\vec{p}\times\vec{q}$ are the scalar and the usual cross product in ${\mathbb R^3}$ respectively.

A $\mathbb H(\R)$-valued function $U$ defined in $\Omega\subset\R^3$ has the representations $U=U_{0}+\vec{U}=\sum_{n=0}^{3}U_ne_n$ with $U_n$ real valued. Properties such as continuity or differentiability have to be understood component wise. 

Let us denote by $\mathbb H(\mathbb C)$ the set of quaternions with complex components instead of real (complex quaternions).

If $q\in\mathbb H(\mathbb{C})$, then $q={\rm Re}\,q+i\,{\rm Im}\,q$, where $i$ is the complex imaginary unit and ${\rm Re}\,q=\sum_{n=0}^3{\rm Re}\,q_ne_n$, ${\rm Im}\,q=\sum_{n=0}^3{\rm Im}\, q_ne_n$ belong to $\mathbb H(\R)$. 

The following first order partial differential operator is called Dirac operator:
\[\mathcal D:=\sum_{n=1}^{3}e_{n}\partial^{1}_{x_n},\]
where $\partial^{1}_{\tau}$ denotes the partial derivative with respect to $\tau$.

Because $-{\mathcal D}{\mathcal D}=\Delta$, Laplacian in $\R^3$, any function which belongs to $\mbox{ker}\,\mathcal D$ is also harmonic.

The Helmholtz operator $\Delta+\kappa^2$ ($\kappa\in\mathbb{C}$) can be factorized as
\begin{equation*}\label{fac}
-(\mathcal D-\kappa)(\mathcal D+\kappa)=\Delta+\kappa^2,
\end{equation*}
as will be clear later, physically $\kappa$ represents the wave number.

For a ${\mathbb H(\mathbb{C})}$-valued function $U$, the displacements of $\mathcal D$ are denoted by  
\[\mathcal D_\kappa\,U:=\mathcal D\,U\mp\kappa\,U=0.\]
The interested reader is referred to \cite{KS, A22} for further information.

\subsection{Fractional vector operations}
In past decades, there has been considerable effort in literature to study boundary problems of pure mathematics and mathematical physics for domains with highly irregular boundaries like non-rectifiable, finite perimeter, fractals and flat chains, see for instance \cite{Borodich} and the references given there. 

In $1992$ Harrison and Norton \cite{HN} presented an approach to the divergence theorem for domains with boundaries of non-integer box dimension. One of the method they employed was the technique introduced by Whitney in \cite{Whitney}, of decomposition of the domain into cubes and extension of functions defined on a closed set to functions defined on the whole of $\mathbb R^3$ (for details in the construction of the Whitney decomposition, we refer to \cite{Stein}). These techniques were also employed by \cite{Aikawa} where an example of uniform domains is given by an open ball minus the centers of Whitney cubes.

Let $W:=\left\{\vec{x}=(x_1, x_2, x_3): a\leq x_1\leq b, \,\,\,\,a\leq x_2\leq b,\,\,\,\,a\leq x_3\leq b\right\}$ be a cube of $\mathbb R^3$. 

The fractional Nabla operator in coordinates $(x_1, x_2, x_3)$ and the quaternionic units $(e_1, e_2, e_3)$ is written as
\begin{equation*}\label{Nabla}
\nabla_W^{\vec{\alpha}}:= e_1{}^C{}\hspace{-0.05cm}D^{\frac{1+\alpha_1}{2}}_{W}[x_1]+ e_2{}^C\hspace{-0.05cm}D^{\frac{1+\alpha_2}{2}}_{W}[x_2]+ e_3{} ^C\hspace{-0.05cm}D^{\frac{1+\alpha_3}{2}}_{W}[x_3], 
\end{equation*}
where $^C{}\hspace{-0.05cm}D^{\frac{1+\alpha_n}{2}}_{W}[x_n]=\,_a^C\hspace{-0.05cm}D^{\frac{1+\alpha_n}{2}}_{x_n}(x_n)$ denotes the left Caputo fractional derivatives with respect to coordinates $x_n$. Here and subsequently $\vec{\alpha}$ stands for the vector $(\alpha_1, \alpha_2, \alpha_3)$ and $0<\alpha_n\leq 1,\,\ n=1,2,3$.

Following the ideas of \cite{Tarasov2008}, we may define the fractional differential operators over cubes $W$ in quaternionic context.

Let $U: W\longrightarrow \mathbb H(\R)$ such that $U_0, U_n \in AC^1[W]$, where $AC^1[W]$ denotes the class of functions such that its respective restrictions to each of the coordinate axes belongs to $AC^1[a,b]$.

\begin{itemize}
\item[(1)] If $U_0=U_0(\vec{x})$, we define its fractional gradient as 
\begin{eqnarray*}\label{Gradiente}
\mbox{Grad}_W^{\vec{\alpha}}U_0&:=&\nabla_W^{\vec{\alpha}}U_0= e_1{}^C{}\hspace{-0.05cm}D^{\frac{1+\alpha_1}{2}}_{W}[x_1]U_0+e_2{} ^C\hspace{-0.05cm}D^{\frac{1+\alpha_2}{2}}_{W}[x_2]U_0+\\&+& e_3{} ^C\hspace{-0.05cm}D^{\frac{1+\alpha_3}{2}}_{W}[x_3]U_0.
\end{eqnarray*}
\item[(2)] If $\vec{U}=\vec{U}(\vec{x})$, then we define its fractional divergence by
\begin{eqnarray}\label{Divergencia}
\mbox{Div}_W^{\vec{\alpha}}\vec{U}&:=&\nabla_W^{\vec{\alpha}}\cdot\vec{U}=\hspace{-0.05cm}^C{}\hspace{-0.05cm}D^{\frac{1+\alpha_1}{2}}_{W}[x_1]U_1+\hspace{0.1cm}^C{}\hspace{-0.05cm}D^{\frac{1+\alpha_2}{2}}_{W}[x_2]U_2+\nonumber\\&+&^C{}\hspace{-0.05cm}D^{\frac{1+\alpha_3}{2}}_{W}[x_3]U_3.
\end{eqnarray}
\item[(3)] The fractional curl operator is defined by
\begin{eqnarray}\label{Rotacional}
\mbox{Curl}_W^{\vec{\alpha}}\vec{U}&:=&\nabla_W^{\vec{\alpha}}\times\vec{U}= e_1\left(^C{}\hspace{-0.05cm}D^{\frac{1+\alpha_2}{2}}_{W}[x_2]U_3-\hspace{0.1cm}^C{}\hspace{-0.05cm}D^{\frac{1+\alpha_3}{2}}_{W}[x_3]U_2\right)\nonumber\\&+&e_2\left(^C{}\hspace{-0.05cm}D^{\frac{1+\alpha_3}{2}}_{W}[x_3]U_1-\hspace{0.1cm}^C{}\hspace{-0.05cm}D^{\frac{1+\alpha_1}{2}}_{W}[x_1]U_3\right)\nonumber\\&+&e_3\left(^C{}\hspace{-0.05cm}D^{\frac{1+\alpha_1}{2}}_{W}[x_1]U_2-\hspace{0.1cm}^C{}\hspace{-0.05cm}D^{\frac{1+\alpha_2}{2}}_{W}[x_2]U_1\right).
\end{eqnarray}
\end{itemize}
Note that these fractional differential operators are non-local and depend on the $W$ cube.

The following relation for fractional differential vector operations is easily adapted from \cite{Tarasov2008}.
\begin{equation}\label{divcurl}
\mbox{Div}_W^{\vec{\alpha}}(\mbox{Curl}_W^{\vec{\alpha}}\vec{U})= 0. 
\end{equation}

A definition of the $3$-parameter fractional Laplace and Dirac operators using left Caputo derivatives can be found in \cite{Ferreira1}.
\begin{equation*}\label{FractionalLaplace}
^C\hspace{-0.05cm}\Delta^{\vec{\alpha}}_{W} := \,^C\hspace{-0.05cm}D^{1+\alpha_1}_{W}[x_1] + \,^C\hspace{-0.05cm}D^{1+\alpha_2}_{W}[x_2] + \,^C\hspace{-0.05cm}D^{1+\alpha_3}_{W}[x_3],
\end{equation*}

\begin{equation*}\label{FractionalDirac}
^C\hspace{-0.03cm}\mathcal D^{\vec{\alpha}}_{W} := e_1 \,^C\hspace{-0.05cm}D^{\frac{1+\alpha_1}{2}}_{W}[x_1]+e_2 \,^C\hspace{-0.05cm}D^{\frac{1+\alpha_2}{2}}_{W}[x_2]+e_3 \,^C\hspace{-0.05cm}D^{\frac{1+\alpha_3}{2}}_{W}[x_3].
\end{equation*}

The fractional Dirac operator $^C\hspace{-0.03cm}\mathcal D^{\vec{\alpha}}_{W}$ factorizes the fractional Laplace operator $^C\hspace{-0.05cm}\Delta^{\vec{\alpha}}_{W}$ for any ${\mathbb H(\mathbb{C})}$-valued function $U={\rm Re}\, U+i\,{\rm Im}\, U$, whenever the components of the functions ${\rm Re}\,U, {\rm Im}\,U$ (its respective restrictions to each of the coordinate axes) satisfying the sufficient conditions presented in Theorem \ref{th2}. As a matter of fact, for such functions we can apply (\ref{Semigrupo}) which together with the multiplication rules of the quaternion algebra and based upon ideas found in \cite[Section 4]{Ferreira1} gives
\begin{equation}\label{Factorizacion}
-{^C\hspace{-0.03cm}\mathcal D^{\vec{\alpha}}_{W}}\left(^C\hspace{-0.03cm}\mathcal D^{\vec{\alpha}}_{W}U\right)=^C\hspace{-0.2cm}\Delta^{\vec{\alpha}}_{W}U.
\end{equation}
We can now state (paraphrasing the Dirac operator case) the fact that the solution of the fractional Dirac operator are fractional harmonic. 

By straightforward calculation we have
\begin{equation}\label{DiracVector}
^C\hspace{-0.03cm}\mathcal D^{\vec{\alpha}}_{W} U= -{\rm Div}_W^{\vec{\alpha}}\vec{U}+{\rm Grad}_W^{\vec{\alpha}}U_{0}+{\rm Curl}_W^{\vec{\alpha}}\vec{U}.
\end{equation}

\section{Fractional physical systems}
In general, physical models can be formulated using the fractional derivatives, where the kernels are interpreted as power-law densities of states, and the fractional order of the derivative corresponds to the physical dimensions of the material \cite{Tarasov2008,Trujillo1}. Moreover, the nonlocality in time and space can be found in phenomena such as the electromagnetism \cite{Baleanu} and the diffusion \cite{Vazquez}.  

In this section we illustrate some examples where the quaternionic fractional approach emerges in linear hydrodynamic and elasticity. These fractional physical systems are motivated by \cite{Zabarankin}; however, the authors did not find in literature the use of quaternionic fractional approach to formulate such systems. 

Let a vector field $\vec{\Phi}=\vec{\Phi}(\vec{x})$ and a scalar field $\Psi_0=\Psi_0(\vec{x})$ related by
\begin{equation}\label{VF1}
\nabla_W^{\vec{\alpha}}\Psi_0+{\rm Curl}_W^{\vec{\alpha}}\vec{\Phi}+(\vec{B}\times \vec{\Phi})+\Psi_0 \vec{A}=0,\,\,\,{\rm Div}_W^{\vec{\alpha}}\vec{\Phi}+\vec{A}\cdot\vec{\Phi}=0, 
\end{equation}
where $\vec{A}, \vec{B}$ are constant real-valued vector and $\vec{x}$ is the position vector. 

For $\vec{A}=0$, (\ref{VF1}) is the generalized Moisil-Teodorescu system, see for instance \cite{Obolasvili}.  

\begin{Example}[{\bf Generalized Moisil-Teodorescu system}]
\begin{equation}\label{VF}
{\rm Curl}_W^{\vec{\alpha}}\vec{\Phi}+(\vec{B}\times \vec{\Phi})=-\nabla_W^{\vec{\alpha}}\Psi_0,\,\,\,{\rm Div}_W^{\vec{\alpha}}\vec{\Phi}=0. 
\end{equation}
Note that for $\vec{B}=0$, (\ref{VF}) is the Moisil-Theodorescu system whereas for $\Psi_0=0$ and $\vec{B}=0$, (\ref{VF}) simplifies to the classical potential flow equations, see for instance \cite{Zabarankin, Moisil, Bitsadze}.
\end{Example}

\begin{Example}[{\bf Moisil-Teodorescu system}]
\begin{equation}\label{Example1}
\begin{cases} 
\,\,\,\,\;\;{\rm Curl}\hspace{0.05cm}^{\vec{\alpha}}_W\vec{\Phi}+\nabla_W^{\vec{\alpha}} \Psi_0=0\cr\,\,\,\,\;\;{\rm Div}\hspace{0.05cm}^{\vec{\alpha}}_W\vec{\Phi}=0.\nonumber
\end{cases}
\end{equation}
\end{Example}

\begin{Example}[{\bf Ideal fluid}]
The velocity field $\vec{\Theta}$ of an ideal fluid is irrotational and incompressible (solenoidal), i.e.
\begin{equation*}\label{Example2}
{\rm Curl}\hspace{0.05cm}^{\vec{\alpha}}_W\vec{\Theta}=0,\,\,\,{\rm Div}\hspace{0.05cm}^{\vec{\alpha}}_W\vec{\Theta}=0, 
\end{equation*}
which corresponds to (\ref{VF}) with $\vec{\Phi}=\vec{\Theta}$, $\Psi_0\equiv 0$ and $\vec{B}=0$.
\end{Example}

\begin{Example}[{\bf Stokes flows}]
Under the assumption of negligible inertial and thermal effects, the time-independent velocity field $\vec{\Theta}$ of a viscous incompressible fluid is governed by the Stokes equations
\begin{equation}\label{Example3}
\mu_0\, ^C\hspace{-0.05cm}\Delta^{\vec{\alpha}}_{W}\vec{\Theta}=\nabla_W^{\vec{\alpha}}P_0,\,\,\,{\rm Div}\hspace{0.05cm}^{\vec{\alpha}}_W\vec{\Theta}=0, 
\end{equation}
where $P_0$ is the pressure in the fluid, $\mu_0$ is shear viscosity and $$^C\hspace{-0.05cm}\Delta^{\vec{\alpha}}_{W}\vec{\Theta}=\nabla_W^{\vec{\alpha}}{\rm Div}\hspace{0.05cm}^{\vec{\alpha}}_W\vec{\Theta}-{\rm Curl}\hspace{0.05cm}^{\vec{\alpha}}_W({\rm Curl}\hspace{0.05cm}^{\vec{\alpha}}_W\vec{\Theta}).$$
The equations (\ref{Example3}) imply that the vorticity $\vec{\Lambda}={\rm Curl}\hspace{0.05cm}^{\vec{\alpha}}_W\vec{\Theta}$ and pressure $P_0$ are related by
\begin{equation*}\label{Example3_1}
\mu_0\, {\rm Curl}\hspace{0.05cm}^{\vec{\alpha}}_W\vec{\Lambda}=-\nabla_W^{\vec{\alpha}}P_0,\,\,\,{\rm Div}\hspace{0.05cm}^{\vec{\alpha}}_W\vec{\Lambda}=0, 
\end{equation*}
which corresponds to (\ref{VF}) with $\Psi_0=P_0$, $\vec{\Phi}=\mu_0\vec{\Lambda}$ and $\vec{B}=0$.
\end{Example}

\begin{Example}[{\bf Oseen flows}]
Suppose a solid body translates with constant velocity $\vec{V}$ in a quiescent viscous incompressible fluid. If the Reynolds number is
sufficiently small, the time-independent velocity field $\vec{\Theta}$ with partially accounted inertial effects can be described by the Oseen equations
\begin{equation}\label{Example4}
\mu_0\, ^C\hspace{-0.05cm}\Delta^{\vec{\alpha}}_{W}\vec{\Theta}+\rho_0(\vec{V}\cdot\nabla_W^{\vec{\alpha}})\vec{\Theta}=\nabla_W^{\vec{\alpha}}P_0,\,\,\,{\rm Div}\hspace{0.05cm}^{\vec{\alpha}}_W\vec{\Theta}=0, 
\end{equation}
where $P_0$ is the pressure, $\mu_0$ and $\rho_0$ are fluid shear viscosity and density, respectively. Let $\vec{V}\cdot\vec{\Lambda}$=0 with $\vec{\Lambda}={\rm Curl}\hspace{0.05cm}^{\vec{\alpha}}_W\vec{\Theta}$. Then (\ref{Example4}) can be recast in two equivalent forms:
\begin{equation}\label{Example4_1}
{\rm Curl}\hspace{0.05cm}^{\vec{\alpha}}_W[\mu_0\vec{\Lambda}+\rho_0(\vec{V}\times\vec{\Theta})]=-\nabla_W^{\vec{\alpha}}P_0,\,\,\,{\rm Div}\hspace{0.05cm}^{\vec{\alpha}}_W[\mu_0\vec{\Lambda}+\rho_0(\vec{V}\times\vec{\Theta})]=0 
\end{equation}
and
\begin{equation}\label{Example4_2}
\mu_0\,{\rm Curl}\hspace{0.05cm}^{\vec{\alpha}}_W\vec{\Lambda}+\rho_0(\vec{V}\times\vec{\Lambda})=-\nabla_W^{\vec{\alpha}}[P_0-\rho_0(\vec{V}\cdot\vec{\Theta})],\,\,\,{\rm Div}\hspace{0.05cm}^{\vec{\alpha}}_W\vec{\Lambda}=0, 
\end{equation}
which are both particular cases of (\ref{VF}): $\Psi_0=P_0$, $\vec{\Phi}=\mu_0\vec{\Lambda}+\rho_0(\vec{V}\times\vec{\Theta})$ and $\vec{B}=0$ in (\ref{Example4_1}), and $\Psi_0=P_0-\rho_0(\vec{V}\cdot\vec{\Theta})$, $\vec{\Phi}=\mu_0 \vec{\Lambda}$ and $\vec{B}=\rho_0 \vec{V}/\mu_0$ in (\ref{Example4_2}). 
\end{Example}

\begin{Example}[{\bf Fractional Lam\'e-Navier system}]
A 3-dimensional field $\vec{U}$ in a homogeneous isotropic linear elastic material without volume forces is 
described by the Lam\'e-Navier system:
\begin{equation}\label{Lame-System}
\mathcal L_{\lambda,\mu}\vec{U}:=\mu\Delta\vec{U}+(\mu+\lambda){\rm grad}({\rm div}\vec{U})=0,
\end{equation}
where $\mu>0$, $\lambda>-\frac{2}{3}\mu$ are the Lam\'e coefficients, see \cite{Moreno1} for more details. 

The fractional calculus can be used to establish a fractional generalization of non-local elasticity in two forms: the fractional gradient elasticity theory (weak non-locality) and the fractional integral elasticity theory (strong non-locality), see \cite{Carpinteri, Aifantis, Aifantis2}.  

Many applications of fractional calculus amount to replacing the spacial derivative in an equation with a derivative of fractional order. So, we can consider a generalization of (\ref{Lame-System}) such that it includes derivatives of non-integer order. 

Precisely, we will propose the following transformations:
\begin{equation}\label{Fractionalization1}
\Delta\longrightarrow\Delta^{\vec{\alpha}}_{W},
\end{equation}

\begin{equation}\label{Fractionalization2}
{\rm grad}\longrightarrow{\rm Grad}_W^{\vec{\alpha}},
\end{equation}

\begin{equation}\label{Fractionalization3}
{\rm div}\longrightarrow{\rm Div}_W^{\vec{\alpha}}.
\end{equation}
Then, we get the fractional Lam\'e-Navier system associated with the transformations (\ref{Fractionalization1})-(\ref{Fractionalization3}) as follows:

\begin{equation}\label{Fractional Lame-System}
\mathcal L_{\lambda,\mu}^{\vec{\alpha}}\vec{U}:=\mu ^C\hspace{-0.05cm}\Delta^{\vec{\alpha}}_{W}\vec{U}+(\mu+\lambda){\rm Grad}_W^{\vec{\alpha}}({\rm Div}_W^{\vec{\alpha}}\vec{U})=0.
\end{equation}

Combining (\ref{divcurl}) with (\ref{DiracVector}) yields  
\begin{equation*}\label{Example5}
(^C\hspace{-0.03cm}\mathcal D^{\vec{\alpha}}_{W})^2\,\vec{U}=-{\rm Grad}_W^{\vec{\alpha}}({\rm Div}_W^{\vec{\alpha}}\vec{U})+{\rm Curl}\hspace{0.05cm}^{\vec{\alpha}}_W ({\rm Curl}\hspace{0.05cm}^{\vec{\alpha}}_W \vec{U}),
\end{equation*}

\begin{equation*}\label{Example5_1}
^C\hspace{-0.03cm}\mathcal D^{\vec{\alpha}}_{W}\,\vec{U}\,^C\hspace{-0.03cm}\mathcal D^{\vec{\alpha}}_{W}=-{\rm Grad}_W^{\vec{\alpha}}({\rm Div}_W^{\vec{\alpha}}\vec{U})-{\rm Curl}\hspace{0.05cm}^{\vec{\alpha}}_W ({\rm Curl}\hspace{0.05cm}^{\vec{\alpha}}_W \vec{U}),
\end{equation*}
and hence we have
\begin{equation*}\label{Example5_2}
{\rm Grad}_W^{\vec{\alpha}}{\rm Div}_W^{\vec{\alpha}}\vec{U}=-\frac{1}{2}\left[(^C\hspace{-0.03cm}\mathcal D^{\vec{\alpha}}_{W})^2\,\vec{U}+\,^C\hspace{-0.03cm}\mathcal D^{\vec{\alpha}}_{W}\vec{U}\,^C\hspace{-0.03cm}\mathcal D^{\vec{\alpha}}_{W}\right].
\end{equation*}
 Consequently, the fractional Lam\'e-Navier system (\ref{Fractional Lame-System}) can be rewritten in the form
\begin{equation*}\label{Fractional Lame-System2}
\frac{(\mu+\lambda)}{2}\,^C\hspace{-0.03cm}\mathcal D^{\vec{\alpha}}_{W}\,\vec{U}\,^C\hspace{-0.03cm}\mathcal D^{\vec{\alpha}}_{W}+\left(\mu+\frac{\mu+\lambda}{2}\right)\left(^C\hspace{-0.03cm}\mathcal D^{\vec{\alpha}}_{W}\right)^2\vec{U}=0.
\end{equation*}
Let us denote $\gamma=\displaystyle\frac{(\mu+\lambda)}{2}$, $\beta=\displaystyle\frac{(3\mu+\lambda)}{2}$ and introduce the operator

\begin{equation}\label{Fractional Lame-System3}
\mathcal L_{\lambda,\mu}^{\ast,\vec{\alpha}}\vec{U}:=\gamma\,^C\hspace{-0.03cm}\mathcal D^{\vec{\alpha}}_{W}\vec{U}\,^C\hspace{-0.03cm}\mathcal D^{\vec{\alpha}}_{W}+\beta(^C\hspace{-0.03cm}\mathcal D^{\vec{\alpha}}_{W})^2\,\vec{U}.
\end{equation} 
Having in mind the conditions relating $\lambda, \mu$ in (\ref{Lame-System}), it is easily seen that $\gamma\neq0$ and $\beta\neq0$.
\end{Example}

\begin{remark}
Note that the operational equation involves (\ref{Fractional Lame-System3}) is equivalent to the fractional Lam\'e-Navier system (\ref{Fractional Lame-System}). 
\end{remark}

\begin{remark}
Observe that, the term $^C\hspace{-0.03cm}\mathcal D^{\vec{\alpha}}_{W}\,\vec{U}\,^C\hspace{-0.03cm}\mathcal D^{\vec{\alpha}}_{W}$ in (\ref{Fractional Lame-System3}) is a generalization of the sandwich equation. Solutions of the sandwich equation ${\mathcal D} \vec{U}{\mathcal D}= 0$ are known as inframonogenic functions, see \cite{Moreno1} for more details.
In this way, the kernel of $^C\hspace{-0.03cm}\mathcal D^{\vec{\alpha}}_{W}\,\vec{U}\,^C\hspace{-0.03cm}\mathcal D^{\vec{\alpha}}_{W}$ could be understood as the set of fractional inframonogenic functions.
\end{remark}

\section{Fractional monochromatic Maxwell system}
The behavior of electric fields $(\E,\D)$, magnetic fields $(\B,\bf H)$, charge density $\rho(t,\vec{x})$, and current density ${\bf j}(t,\vec{x})$ is described by the Maxwell system, see \cite{NA48} and the references given there.

The relations between electric fields $(\E,\D)$ for the medium can be realized by the convolution
\begin{equation}\label{relacionED}
\D(t,\vec{x})=\varepsilon_0\int_{-\infty}^{+\infty}\varepsilon(\vec{x},\acute{\vec{x}})\E(t,\acute{\vec{x}})d\acute{\vec{x}},
\end{equation}
where $\varepsilon_0$ is the permittivity of free space. Homogeneity in space gives $\varepsilon(\vec{x},\acute{\vec{x}})=\varepsilon(\vec{x}-\acute{\vec{x}})$. A local case accords with the Dirac delta-function permittivity $\varepsilon(\vec{x})=\varepsilon\delta(\vec{x})$ and (\ref{relacionED}) yields $\D(t,\vec{x})=\varepsilon_0\varepsilon\E(t,\vec{x})$.

Analogously, we have a non-local equation for the magnetic fields $(\B,\bf H)$.

\subsection{Fractional non-local Maxwell system}
A feasible way of appearance of the Caputo derivative in the classical electrodynamics can be found in \cite{Tarasov2008}. This is mainly included here to keep the exposition self-contained.

If we have 
\begin{equation*}\label{relacionED1}
\D(t,x_1)=\int_{-\infty}^{+\infty}\varepsilon(x_1-\acute{x_1})\E(t,\acute{x_1})d\acute{x_1},
\end{equation*}        
then
\begin{eqnarray*}
\partial^{1}_{x_1}\D(t,x_1)&=&\int_{-\infty}^{+\infty}\partial^{1}_{x_1}\varepsilon(x_1-\acute{x_1})\E(t,\acute{x_1})d\acute{x_1}\\&=&-\int_{-\infty}^{+\infty}\partial^{1}_{\acute{x_1}}\varepsilon(x_1-\acute{x_1})\E(t,\acute{x_1})d\acute{x_1}.
\end{eqnarray*}

The integration by parts now leads to
\begin{equation}\label{relacionED2}
\partial^{1}_{x_1}\D(t,x_1)=\int_{-\infty}^{+\infty}\varepsilon(x_1-\acute{x_1})\partial^{1}_{\acute{x_1}}\E(t,\acute{x_1})d\acute{x_1}.
\end{equation} 
The non-local properties of electrodynamics can be considered as a result of dipole-dipole interactions with a fractional power-law screening that is connected with the integro-differentiation of non-integer order, see \cite{Tarasov2008_2}.

Consider the kernel $\varepsilon(x_1-\acute{x_1})$ of (\ref{relacionED2}) in $(0,x_1)$ such that
\begin{equation*}\label{kernel}
\varepsilon(x_1-\acute{x_1}) =
\begin{cases} 
e(x_1-\acute{x_1}),\,\,0<\acute{x_1}<x_1,\cr 0,\;\;\;\;\;\;\;\;\;\;\;\;\acute{x_1}>x_1,\,\,\, \acute{x_1}<0,
\end{cases}
\end{equation*} 
with the power-like function
\begin{equation*}\label{kernel1}
e(x_1-\acute{x_1}) = \frac{1}{\Gamma(1-\alpha_1)}\frac{1}{(x_1-\acute{x_1})^{\alpha_1}}, \,\,\,\,\, (0<\alpha_1<1).
\end{equation*} 
Then (\ref{relacionED2}) gives the relation 
\begin{equation*}\label{relacionfraccED}
\partial^{1}_{x_1}\D(t,x_1)= _0^C\hspace{-0.15cm}D^{\alpha_1}_{x_1}\,\E(t,x_1), \,\,\,\,\, (0<\alpha_1<1),
\end{equation*}
with the Caputo fractional derivatives $_0^C\hspace{-0.05cm}D^{\alpha_1}_{x_1}$.

Let us apply (\ref{Divergencia}) and (\ref{Rotacional}) to write the corresponding fractional non-local differential Maxwell system as
\begin{equation}\label{m3}
\begin{cases} 
\,\,\,\,\;\;\mbox{Div}_W^{\vec{\alpha}}\E(t,\vec{x})=g_1\rho(t,\vec{x})\cr\,\,\,\,\;\;\mbox{Curl}_W^{\vec{\alpha}}\E(t,\vec{x})=-\partial^{1}_{t} \B(t,\vec{x})\cr \,\,\,\,\;\;\mbox{Div}_W^{\vec{\alpha}}\B(t,\vec{x})=0\cr \,\,\,\,\;\;g_2\mbox{Curl}_W^{\vec{\alpha}}\B(t,\vec{x})={\bf j}(t,\vec{x}) + g_{3}^{-1}\partial^{1}_{t}\E(t,\vec{x}),
\end{cases}
\end{equation}
where $g_1$, $g_2$, $g_3$ are constants. We assume that the densities $\rho(t,\vec{x})$ and ${\bf j}(t,\vec{x})$, which describe the external sources of the electromagnetic field, are given.
 
The main idea behind the use of fractional differentiation, for describing real-world problems, is their abilities to describe non-local distributed effects. For example, a power-law long-range interaction in the 3-dimensional lattice in the continuous limit can give a fractional equation, see \cite{Tarasov2006, Tarasov2006_2}. In \cite{Aguilar}, some numerical examples and simulations are provided to illustrate the use of alternative fractional differential equations for modeling the electrical circuits. 

Also, the methodology used in \cite{Aguilar} succeed in the analysis of electromagnetic transients problems in electrical systems. Moreover, an empirical model for complex permittivity was incorporated into Maxwell’s equations that lead to the appearance of fractional order derivatives in Ampere’s Law and the wave equation, see \cite{Wharmby}. 

The fractional Maxwell system (\ref{m3}) can describe electromagnetic fields in media with fractional non-local properties, like in superconductor and semi-conductor physics \cite{Belleguie, Genchev} and in accelerated systems \cite{Mashhoon}.

\subsection{Fractional monochromatic Maxwell system}
We will assume that the electromagnetic characteristics of the medium do not change in time. If in addition they have the same values in each point of the cube $W\in\R^3,$ then the medium which fills the volume is called homogeneous.   

A monochromatic electromagnetic field has the following form
\begin{equation}\label{E}
\E(t,\vec{x})= \mbox{Re}(\vec{E}(\vec{x})e^{-i\omega t})
\end{equation}
and 

\begin{equation}\label{B}
{\bf B}(t,\vec{x})= \mbox{Re}(\vec{B}(\vec{x})e^{-i\omega t}),
\end{equation}
where $\vec{E}: W\longrightarrow {\mathbb R^3}$, $\vec{B}: W\longrightarrow {\mathbb R^3}$ and all dependence on time is contained in the factor $e^{-i\omega t}$. 

$\vec{E}$ and $\vec{B}$ are complex vectors called the complex amplitudes of the electromagnetic field; $\omega$  is the frequency of oscillations.

Substituting (\ref{E}) and (\ref{B}) into (\ref{m3}) we obtain the equations for the complex amplitudes $\vec{E}$ and $\vec{B}$:
\begin{equation}\label{m5}
\begin{cases} 
\,\,\,\,\;\;\mbox{Div}_W^{\vec{\alpha}}\vec{E}=g_1\rho\cr\,\,\,\,\;\;\mbox{Curl}_W^{\vec{\alpha}}\vec{E}=i \omega\vec{B}\cr \,\,\,\,\;\;\mbox{Div}_W^{\vec{\alpha}}\vec{B}=0\cr \,\,\,\,\;\;\mbox{Curl}_W^{\vec{\alpha}}\vec{B}=-i\omega g^{-1}_2g^{-1}_3\vec{E}+g^{-1}_2\vec{j}.
\end{cases}
\end{equation}
The quantities $\rho$ and ${\bf j}$ are also assumed to be monochromatic $\rho(\vec{x},t)= \mbox{Re}(\rho(\vec{x})e^{-i\omega t})$, ${\bf j}(\vec{x},t)= \mbox{Re}(\vec{j}(\vec{x})e^{-i\omega t})$.

\subsection{Fractional Helmholtz operator}
The following fractional wave equations can be found in \cite{Tarasov2008}. Using the fractional non-local Maxwell system with ${\bf j}=0$ and $\rho=0$, are obtained the wave equations for electric and magnetic fields in a $W$ cube. 
\begin{equation}\label{wave1}
\begin{cases} 
\,\,\,\,\;\;\displaystyle\frac{1}{\nu^{2}}\partial^{2}_{t}{\bf B}- ^C\hspace{-0.2cm}\Delta^{\vec{\alpha}}_{W}{\bf B}=0\cr\,\,\,\,\;\;\displaystyle\frac{1}{\nu^{2}}\partial^{2}_{t}{\bf E}- ^C\hspace{-0.2cm}\Delta^{\vec{\alpha}}_{W}{\bf E}=0,
\end{cases}
\end{equation}
where $\nu^{2}=g_{2}g_{3}$.

Substituting (\ref{E}) and (\ref{B}) into (\ref{wave1}), we obtain that $\vec{B}$, $\vec{E}$ are also solutions of homogeneous fractional Helmholtz equations with respect to the square of the medium wave number $\displaystyle\frac{\omega}{\nu}$.
\begin{equation}\label{Helmholtz1}
\begin{cases} 
\,\,\,\,\;\;\displaystyle\frac{\omega^2}{\nu^{2}}\vec{B}- ^C\hspace{-0.2cm}\Delta^{\vec{\alpha}}_{W}\vec{B}=0\cr\,\,\,\,\;\;\displaystyle\frac{\omega^2}{\nu^{2}}\vec{E}- ^C\hspace{-0.2cm}\Delta^{\vec{\alpha}}_{W}\vec{E}=0.
\end{cases}
\end{equation}
The above fractional Helmholtz equations motivate the introduction of the fractional Helmholtz operator $^C\hspace{-0.05cm}\Delta^{\vec{\alpha}}_{W}+\kappa^2$ $\left(\kappa=\displaystyle\frac{\omega}{\nu}\in\mathbb C\right)$.

Next, the fractional Helmholtz operator, can be factorized as
\begin{equation}\label{facfracc}
-\left(^C\hspace{-0.05cm}\mathcal D^{\vec{\alpha}}_{W}-\kappa\right)\left(^C\hspace{-0.05cm}\mathcal D^{\vec{\alpha}}_{W}+\kappa\right)=^C\hspace{-0.2cm}\Delta^{\vec{\alpha}}_{W}+\kappa^2,
\end{equation}
which is a corollary of (\ref{Factorizacion}).

The formulation of (\ref{m5}) and (\ref{Helmholtz1}) in terms of the fractional Dirac operator $^C\hspace{-0.03cm}\mathcal D^{\vec{\alpha}}_{W}= -{\rm Div}_W^{\vec{\alpha}}+{\rm Curl}_W^{\vec{\alpha}}$ allows us to describe solutions of both systems in terms of displacements $^C\hspace{-0.03cm}\mathcal D^{\vec{\alpha}}_{W}\mp\kappa$. 

\begin{remark}
In view of the factorization (\ref{facfracc}) of the fractional Helmholtz operator, we can express the solutions $\vec{E}$ and $\vec{B}$ of (\ref{Helmholtz1}) in terms of the function $\vec{F}=\left(^C\hspace{-0.05cm}\mathcal D^{\vec{\alpha}}_{W}+\kappa\right)[\vec{E}+i\vec{B}]$. That function belongs to $\mbox{ker}\,\left(^C\hspace{-0.05cm}\mathcal D^{\vec{\alpha}}_{W}-\kappa\right)$ and in turn allows us to re-express the electric and magnetic components for (\ref{m5}). 
\end{remark}

Applying the fractional divergence operator to the last equation in (\ref{m5}) and using (\ref{divcurl}), we find the relation between $\rho$ and $\vec{j}$:
\begin{equation}\label{relation}
\mbox{Div}_W^{\vec{\alpha}}\vec{j}-i\omega\rho g_1g^{-1}_{3}=0.
\end{equation}
In order to rewrite (\ref{m5}) in quaternionic form, let us denote the wave number $\kappa:=\omega\sqrt{g^{-1}_2g^{-1}_3}$, where the square root is chosen so that $\mbox{Im}\,\kappa\geq0$.

Introduce the following pair of purely vector biquaternionic functions 
\begin{equation}\label{BQF1}
\vec{\varphi}:=-i\omega g^{-1}_2g^{-1}_3\vec{E}+\kappa\vec{B}, 
\end{equation}

\begin{equation}\label{BQF2}
\vec{\psi}:=i\omega g^{-1}_2g^{-1}_3\vec{E}+\kappa\vec{B}
\end{equation}
and the notation
\begin{equation*}\label{FPDO}
^C\hspace{-0.03cm}\mathcal D^{\vec{\alpha}}_{W,\kappa}:=^C\hspace{-0.15cm}\mathcal D^{\vec{\alpha}}_{W}+\kappa.
\end{equation*}
Now, we formulate the main result of this paper, which consist of a biquaternionic reformulation of a fractional monochromatic Maxwell system.   
\begin{Theorem}
The fractional quaternionic equations
\begin{equation}\label{equivalent_equation1}
^C\hspace{-0.03cm}\mathcal D^{\vec{\alpha}}_{W,-\kappa}\vec{\varphi}=g^{-1}_2\left({\rm Div}_W^{\vec{\alpha}}\vec{j}+\kappa \vec{j}\right),
\end{equation}
\begin{equation}\label{equivalent_equation2}
^C\hspace{-0.03cm}\mathcal D^{\vec{\alpha}}_{W,\kappa}\vec{\psi}=g^{-1}_2\left(-{\rm Div}_W^{\vec{\alpha}}\vec{j}+\kappa \vec{j}\right).
\end{equation} 
are equivalent to the fractional Maxwell system (\ref{m5}). Indeed, $\vec{\varphi}$ and $\vec{\psi}$ are solutions of (\ref{equivalent_equation1}) and (\ref{equivalent_equation2}) respectively, if and only if $\vec{E}$ and $\vec{B}$ are solutions of (\ref{m5}).  
\end{Theorem}
\begin{proof}
Let $\vec{E}$ and $\vec{B}$ solutions of (\ref{m5}), which can be rewritten as two quaternionic equations
\begin{equation}\label{D1}
^C\hspace{-0.03cm}\mathcal D^{\vec{\alpha}}_{W}\vec{E}=i\omega\vec{B}-g_1\rho,  
\end{equation}
\begin{equation}\label{D2}
^C\hspace{-0.03cm}\mathcal D^{\vec{\alpha}}_{W}\vec{B}=-i\omega g^{-1}_2g^{-1}_3\vec{E}+g^{-1}_2\vec{j}.  
\end{equation}
Applying $^C\hspace{-0.03cm}\mathcal D^{\vec{\alpha}}_{W}$ to $\vec{\varphi}$ in (\ref{BQF1}) and combining (\ref{D1}) with (\ref{D2}) we get 
\begin{eqnarray}\label{Dphi}
^C\hspace{-0.03cm}\mathcal D^{\vec{\alpha}}_{W}\vec{\varphi}&=&-g^{-1}_{2}g^{-1}_{3}i\omega\, ^C\hspace{-0.03cm}\mathcal D^{\vec{\alpha}}_{W}\vec{E}+\kappa \,^C\hspace{-0.03cm}\mathcal D^{\vec{\alpha}}_{W}\vec{B}\nonumber\\
&=&\kappa^2\vec{B}+g_1g^{-1}_{2}g^{-1}_{3}i\omega\rho-\kappa i\omega g^{-1}_2g^{-1}_3\vec{E}+\kappa g^{-1}_2\vec{j}\nonumber\\
&=&\kappa\vec{\varphi}+g_1g^{-1}_{2}g^{-1}_{3}i\omega\rho+\kappa g^{-1}_2\vec{j}\nonumber. 
\end{eqnarray}
Thus, (\ref{relation}) shows that $\vec{\varphi}$ satisfies (\ref{equivalent_equation1}). Analogously we can assert that $\vec{\psi}$ in (\ref{BQF2}) satisfies (\ref{equivalent_equation2}).

On the contrary, suppose that $\vec{\varphi}$ and $\vec{\psi}$ satisfies (\ref{equivalent_equation1}) and (\ref{equivalent_equation2}) respectively. A trivial verification shows that 
\begin{equation}\label{Dphi2}
^C\hspace{-0.03cm}\mathcal D^{\vec{\alpha}}_{W}\vec{\varphi}=\kappa\vec{\varphi}+g_1g^{-1}_{2}g^{-1}_{3}i\omega\rho+\kappa g^{-1}_2\vec{j}.
\end{equation}

Substituting (\ref{BQF1}) into (\ref{Dphi2})
\begin{eqnarray*}\label{Dphi3}
-g^{-1}_{2}g^{-1}_{3}i\omega\, ^C\hspace{-0.03cm}\mathcal D^{\vec{\alpha}}_{W}\vec{E}+\kappa \,^C\hspace{-0.03cm}\mathcal D^{\vec{\alpha}}_{W}\vec{B}\nonumber=\kappa^2\vec{B}+g_1g^{-1}_{2}g^{-1}_{3}i\omega\rho-\kappa i\omega g^{-1}_2g^{-1}_3 \vec{E}+\kappa g^{-1}_2 \vec{j}\\=-i^2\omega^2 g^{-1}_2 g^{-1}_3\vec{B}+g_1g^{-1}_{2}g^{-1}_{3}i\omega\rho-\kappa i\omega g^{-1}_2g^{-1}_3\vec{E}+\kappa g^{-1}_2 \vec{j}\nonumber\\=-i\omega g^{-1}_2 g^{-1}_3(i\omega\vec{B}-g_1\rho)+\kappa(-i\omega g^{-1}_2g^{-1}_3\vec{E}+g^{-1}_2\vec{j}).
\end{eqnarray*}

From the last equality, we conclude that (\ref{D1}) holds. Similar considerations apply to $\vec{\psi}$ in order to obtain (\ref{D2}).

Separating the vector and scalar parts in (\ref{D1}) and (\ref{D2}), together with the vectorial nature of $\vec{\varphi}$, $\vec{\psi}$ and (\ref{DiracVector}) yields (\ref{m5}). This completes the proof.
\end{proof}

\section{Conclusions}
The main purpose of this paper was to explored the very close connection between the $3$-parameter quaternionic displaced fractional Dirac operator with a fractional monochromatic Maxwell system using Caputo derivatives. With this aim in mind, a biquaternionic reformulation of such a system was studied. Moreover, some examples to illustrate how the quaternionic fractional approach emerges in linear hydrodynamic and elasticity are given. As future works, the formulation of a fractional inframonogenic functions theory is suggested.

\subsection*{Data Availability}
No data were used to support this study.

\subsection*{Conflicts of Interest}
The authors declare that they have no conflicts of interest.

\subsection*{Funding Statement}
Juan Bory-Reyes was partially supported by Instituto Polit\'ecnico Nacional in the framework of SIP programs (SIP20195662).

\subsection*{Acknowledgements}
Yudier Pe\~na-P\'erez gratefully acknowledges the financial support of the Postgraduate Study Fellowship of the Consejo Nacional de Ciencia y Tecnolog\'ia (CONACYT) (grant number 744134).

\subsection*{Authors Contributions} 
All authors contributed equally to this work. All authors read and approved the final manuscript.

\end{document}